\documentclass{fundam}

\usepackage{amsmath,graphicx, amsfonts, amssymb}
\usepackage{mathrsfs}
\usepackage{enumerate}
\usepackage[shortlabels]{enumitem}
\usepackage{graphicx}
\usepackage{mathtools}
\usepackage{hyperref}
\usepackage{tikz}
\usetikzlibrary{calc, shapes, positioning}

\newcommand{\Cc}{\mathscr{C}}
\newcommand{\N}{\mathbb{N}}

\newcommand{\nc}[1]{\newcommand{#1}}
\newcommand{\rnc}[1]{\renewcommand{#1}}

\rnc{\O}{{\cal O}}

\nc{\w}{\overline{w}}
\nc{\x}{\overline{x}}
\nc{\y}{\overline{y}}
\nc{\z}{\overline{z}}
\rnc{\a}{\overline{a}}
\rnc{\b}{\overline{b}}
\rnc{\c}{\overline{c}}
\rnc{\d}{\overline{d}}

\newtheorem{theorem}{Theorem}[section]

\usepackage[noabbrev,capitalise,nameinlink]{cleveref}

\crefname{lemma}{Lemma}{Lemmas}
\crefname{theorem}{Theorem}{Theorems}
\begin{document}
\def\figurename{Figure}
%%%%%%%  parameters to be filled in by copy-editor  %%%%%%%%%%

\setcounter{page}{129}
\publyear{24}
\papernumber{2175}
\volume{191}
\issue{2}

\finalVersionForARXIV
%\finalVersionForIOS
%%%%%%%%%%%%%%%%%%%%%%%%%%%%%%%%%%

\title{Elimination Distance to Bounded Degree  on \\ Planar Graphs Preprint}

\author{Alexander Lindermayr, \ Sebastian Siebertz, \ Alexandre Vigny\thanks{Address for
                  correspondence: University of Bremen, Bibliothekstra{\ss}e 1, 28359 Bremen, Germany.  \newline \newline
                    \vspace*{-6mm}{\scriptsize{Received January 2022; \ accepted April 2024.}}}
\\
University of Bremen\\
Germany\\
\{linderal, siebertz, vigny\}@uni-bremen.de
}

\runninghead{A. Lindermayr et al.}{Elimination Distance to Bounded Degree on Planar Graphs Preprint}

\maketitle

\vspace*{-4mm}
\begin{abstract}
  We study the graph parameter \emph{elimination distance to bounded
    degree}, which was introduced by Bulian and Dawar in their study
  of the parameterized complexity of the graph isomorphism problem. We
  prove that the problem is fixed-parameter tractable on planar
  graphs, that is, there exists an algorithm that given a planar graph
  $G$ and integers $d,k$ decides in time $f(k,d)\cdot n^c$ for a
  computable function~$f$ and constant $c$
  whether the elimination distance of
  $G$ to the class of degree $d$ graphs is at most $k$.
\end{abstract}

% !TEX root = main-arxiv.tex

\section{Introduction}

Structural graph theory offers a wealth of parameters that measure the
complexity of graphs or graph classes. Among the most prominent
parameters are \emph{treedepth} and \emph{tree\-width}, which
intuitively measure the resemblance of graphs with stars and trees,
respectively. Other commonly studied structurally restricted graph
classes are the class of \emph{planar graphs}, classes that exclude a
fixed graph as a \emph{minor} or \emph{topological minor}, classes of
\emph{bounded expansion} and \emph{nowhere dense} classes.

%\medskip
Once we have gained a good understanding of a graph class $\Cc$, it is
natural to study classes whose members are \emph{close} to graphs in
$\Cc$. One of the simplest measures of distance to a graph class $\Cc$
is the number of vertices or edges that one must delete (or add) to a
graph~$G$ to obtain a graph from $\Cc$. Guo et
al.~\cite{guo2004structural} formalized this concept under the name
\emph{distance from triviality}. For example, the size of a vertex
cover is the distance to the class of edgeless graphs and the size of
a feedback vertex set is the distance to the class of acyclic graphs. More
generally, for a graph~$G$ a vertex set $X$ is called a
\emph{$c$-treewidth modulator} if the treewidth of~$G-X$ is at most
$c$, hence, the size of a $c$-treewidth modulator corresponds to the
distance to the class of graphs of treewidth at most $c$. This concept
was introduced and studied by Gajarsk\'y et al.\
in~\cite{gajarsky2017kernelization}.

The \emph{elimination distance} to a class $\Cc$ of graphs measures
the number of recursive deletions of vertices needed for a graph $G$
to become a member of $\Cc$. More precisely, a graph $G$ has
elimination distance $0$ to $\Cc$ if $G \in \Cc$, and otherwise
elimination distance~$k+1$, if in every connected component of $G$ we
can delete a vertex such that the resulting graph has elimination
distance $k$ to $\Cc$. Elimination distance was introduced by Bulian
and Dawar~\cite{bulian2016graph} in their study of the parameterized
complexity of the graph isomorphism problem. Elimination distance
naturally generalizes the concept of treedepth, which corresponds to
the elimination distance to the class $\Cc_0$ of edgeless graphs.

Elimination distance has very nice algorithmic applications. On the
one hand, small elimination distance to a class $\Cc$ on which
efficient algorithms for certain problems are known to exist, may
allow to lift the applicability of these algorithms to a larger class
of graphs. For example, Bulian and Dawar~\cite{bulian2016graph} showed
that the graph isomorphism problem is fixed-parameter tractable when
parameterized by the elimination distance to the class~$\Cc_d$ of
graphs with maximum degree bounded by $d$, for any fixed integer $d$.
Recently, Hols \emph{et al.}~\cite{DBLP:journals/siamdm/HolsKP22} proved the
existence of polynomial kernels for the vertex cover problem
parameterized by the size of a deletion set to graphs of bounded
elimination distance to different classes of graphs. The related
concept of recursive backdoors has recently been studied in the
context of efficient SAT solving~\cite{dreier2024sat,MahlmannSV21}.

On the other hand, it is an interesting algorithmic question by itself
to determine the elimination distance of a given graph $G$ to a class
$\Cc$ of graphs. It is well known (see
e.g.~\cite{bodlaender1998rankings, nevsetvril2012sparsity,
 reidl2014faster}) that computing treedepth,
i.e.\ elimination distance to $\Cc_0$, is fixed-parameter tractable.
More precisely, we can decide in time $f(k)\cdot n$ whether an $n$-vertex
graph $G$ has treedepth at most $k$. Bulian and Dawar
proved in~\cite{bulian2017fixed} that computing the elimination
distance to any minor-closed class~$\Cc$ is fixed-parameter
tractable when parameterized by the elimination distance.
They also raised the question
whether computing the elimination distance to the class $\Cc_d$ of graphs
with maximum degree at most $d$ is fixed-parameter
tractable when parameterized by the elimination distance and $d$. Note
that this question is not answered by their result for minor-closed
classes, since $\Cc_d$
is not closed under taking~minors.

For $k,d\in \N$, we denote by $\Cc_{k,d}$ the class of all
graphs that have elimination distance at most~$k$ to~$\Cc_d$. It is
easy to see that for every fixed $k$ and $d$ we can formulate the
property that a graph is in~$\Cc_{k,d}$ by a sentence in monadic
second-order logic (MSO). By the famous theorem of
Courcelle~\cite{courcelle1990monadic} we can test every MSO-property
$\varphi$ in time $f(|\varphi|, t)\cdot n$ on every $n$-vertex graph
of treewidth $t$ for some computable function $f$. Hence, we can decide
for every $n$-vertex graph $G$ of treewidth $t$ 
whether~$G\in~\Cc_{k,d}$ in time $f(k,d,t)\cdot n$ for some computable
function $f$. However, for $d \geq 3$ already the class $\Cc_d$ has
unbounded treewidth, and so the same holds for $\Cc_{k,d}$ for all
values of $k$. Thus, Courcelle's Theorem cannot be directly applied to derive
fixed-parameter tractability of the problem in full generality.

On the other hand, it is easy to see that the graphs in $\Cc_{k,d}$ exclude the complete graph $K_{k+d+2}$ as a topological minor, and hence,
for every fixed $k$ and $d$, the class~$\Cc_{k,d}$ in particular has
bounded expansion and is nowhere dense. %We can efficiently test
%for topological minors~\cite{grohe2011topominors} and
We can
efficiently test
first-order~(FO) properties on bounded expansion and nowhere dense classes~\cite{dvovrak2013testing,grohe2017deciding}, however, first-order
logic is too weak to express the elimination distance
problem. This follows from the
fact that first-order logic is too weak to express even connectivity of
a graph or to define connected components.

While we are unable to resolve the question of
Bulian and Dawar in full
generality, in this work we initiate the quest of determining the
parameterized complexity of elimination distance to bounded degree
graphs for restricted classes of inputs. We prove that for every
$n$-vertex graph~$G$
that excludes~$K_5$ as a minor (in particular for every planar graph)
we can test whether $G\in \Cc_{k,d}$ in time $f(k,d)\cdot n^c$ for a
computable function $f$ and constant $c$. Hence, the problem is fixed-parameter
tractable with parameters $k$ and $d$ when restricted to $K_5$-minor-free graphs.

\begin{theorem}[Main result]\label{thm-main}
  There is an algorithm that for a $K_5$-minor-free input graph~$G$
  with $n$ vertices and integers $k,d$, decides in time
  $f(k,d)\cdot n^c$ whether $G$ belongs to
  $\Cc_{k,d}$, where $f$ is a computable
  function and $c$ is a constant.

\medskip
  In case $G$ is $K_5$- and $K_{3,3}$-minor-free (that is, $G$ is planar), the running time is $f(k,d)\cdot n^3$ for a computable
   function~$f$.
\end{theorem}

Observe that the result is not implied by the result of
Bulian and Dawar for minor-closed classes, as the $K_5$-minor-free
subclass of $\Cc_{d}$ is not minor-closed. It is natural to consider
as a next step classes that exclude some fixed graph as a minor or as
a topological minor, and finally to resolve the problem in full
generality.

\medskip
 To solve the problem on $K_5$-minor-free graphs we combine multiple
techniques from parameterized complexity theory and structural graph
theory. First, we use the fact that the property of having elimination
distance at most $k$ to $\Cc_d$ for fixed $k$ and $d$ is MSO definable, and hence
efficiently solvable by Courcelle's Theorem on graphs of bounded
treewidth. If the input graph $G$ has small treewidth, we can hence
solve the instance by Courcelle's Theorem.
If $G$ has large treewidth, we conclude that it
contains a large grid minor~\cite{robertson1986graph}. We then have two cases.

\medskip
First, if in sufficiently many branch sets there are vertices of degree exceeding
$d$, we conclude that the graph does not belong to $\Cc_{k,d}$. This is because
the grid minor will not be sufficiently scattered by~$k$ elimination rounds.

Second, if only few branch sets contain vertices of degree exceeding $d$, we can
find a branch set that is ``far'' from any problematic vertices. The vertices of
such branch set are called {\em irrelevant}, that is, vertices whose
deletion does not change containment in $\Cc_{k,d}$. By iteratively
removing irrelevant vertices we arrive either in the first case, or at an
instance of small treewidth. In both cases, we can conclude.
The irrelevant vertex technique was introduced
in~\cite{robertson1995graph} and is by now a standard technique in
parameterized algorithms, see~\cite{thilikos2012graph} for a survey.

\paragraph{Recent advances.}

We first presented \cref{thm-main} in a conference paper~\cite{lindermayr2020elimination}.
In this journal version, the presentation of the proof has been vastly
improved and simplified. While the main result is the same,
some observations enabled us to remove unnecessary complicated notion
and to simplify the proof.

Since the first publication of our results much further work has been done
on algorithms computing elimination distances.
Most notably, Agrawal et al.~\cite{DBLP:journals/siamdm/AgrawalKPRS22} establishing fixed-parameter tractability of the problem on general graphs.
In fact, they showed that their approach yields an efficient algorithm to decide
the elimination distance to any graph class that can be described by a finite
set of forbidden induced subgraphs. In a second work
it was shown that there exists a FPT algorithm for computing elimination distance
to classes characterized by the exclusion of a family of finite graphs as
topological minors~\cite{DBLP:journals/corr/abs-2104-09950}. In another line of
research, Agrawal and Ramanujan initiated the study of computing elimination distance to the class of cluster graphs~\cite{DBLP:conf/iwpec/Agrawal020}.

Very recently, Diner et al.\ introduced in~\cite{DBLP:journals/gc/DinerGST22} the notion of block elimination distance, which
replaces the connectivity property in Bulian and Dawar's classical notion of elimination distance with biconnectivity,
while Fomin, Golovach and Thilikos further generalized several of these prior results by considering elimination distances to graph classes
expressible by a restricted first-order logic formulas~\cite{DBLP:journals/tocl/FominGT22}. The results are also implied by the very recent algorithmic
meta theorem for separator logic~\cite{DBLP:conf/icalp/PilipczukSSTV22,DBLP:journals/tocl/SchirrmacherSV23}.

\section{Preliminaries}\label{sec:prelims}

A \emph{graph} $G$ consists of a set of vertices $V(G)$ and a set of edges
$E(G)$. We assume that graphs are finite, simple
and undirected, and we write $\{u,v\}$ for an edge between the vertices
$u$ and $v$. For a set of vertices $S \subseteq V(G)$, we denote the
subgraph of $G$ induced by the vertices $V(G) \setminus S$ by $G - S$.
If~$S = \{a\}$, we write $G - a$.
The \emph{degree} of a vertex~$v$ is the number of edges $e$ such that $v \in e$. The maximum degree of a graph is the largest degree of its vertices.

\medskip
A \emph{partial order}
on a set $V$ is a binary relation $\leq$ on $V$ that is reflexive,
anti-symmetric and transitive. A set $W\subseteq V$ is a \emph{chain}
if it is totally ordered by $\leq$.
If $\leq$ is a partial order on~$V$, and for every element $v\in V$ the set
$V_{\leq v}\coloneqq \{u\in V\mid u\leq v\}$ is a chain, then $\leq$
is a \emph{tree order}.
Note that the covering relation of a tree
order is not necessarily a tree, but may be a forest. An
\emph{elimination order} on a graph $G$ is a tree order $\leq$ on
$V(G)$ such that for every edge $\{u,v\}\in E(G)$ we have either
$u\leq v$ or $v\leq u$. The \emph{depth} of a vertex~$v$ in a tree order~$\leq$ is the size the set $V_{<v}\coloneqq \{u\in V\mid u<v\}$.
The \emph{depth} of a tree order $\leq$ is maximal depth among all vertices.

\medskip
The \emph{treedepth} of a graph $G$ is defined recursively as follows.
%\vspace{-4mm}
\[\mathrm{td}(G)=\begin{cases}
  0 & \text{if $G$ is edgeless,}\\
  1+\min\{\mathrm{td}(G-v)\mid v\in V(G)\} & \text{if $G$ is connected and}\\
  & \hspace{3.18cm}\text{
  not edgeless,}\\
  \max\{\mathrm{td}(H)\mid \text{$H$ connected component of
    $G$}\}&\text{otherwise}.
\end{cases}\]

A graph $G$ has treedepth at most $k$ if and only if there exists an
elimination order on~$G$ of depth at most $k$. If the longest path
in $G$ has length $k$, then its treedepth is bounded by $k$ and
an elimination order of at most this depth can be found in linear time by a depth-first-search.

\smallskip
Elimination distance to a class $\Cc$ naturally generalizes the
concept of treedepth. Let~$\Cc$ be a class of graphs. The
\emph{elimination distance} of $G$ to $\Cc$ is defined recursively as
%\vspace{-4mm}
\[\mathrm{ed}_\Cc(G)=\begin{cases}
  0 & \text{if $G\in \Cc$,}\\
  1+\min\{\mathrm{ed}_\Cc(G-v)\mid v\in V(G)\} & \text{if $G\not\in \Cc$ and $G$ is connected,}\\
  \max\{\mathrm{ed}_\Cc(H)\mid \text{$H$ connected component of
    $G$}\}&\text{otherwise}.
\end{cases}\]

We denote by $\Cc_d$ the class of all graphs of maximum degree at most
$d$ and by $\Cc_{k,d}$ the class of all graphs with elimination
distance at most $k$ to $\Cc_d$. Note for instance that
$\mathrm{td}(G) =k$ if and only if~$G\in~\Cc_{k,0}$.
We write $\mathrm{ed}_d(G)$ for
$\mathrm{ed}_{\Cc_d}(G)$.

\begin{definition}[Definition 4.2 of~\cite{bulian2016graph}]\label[definition]{def:elim-order}
  A tree order $\le$ on G is an {\em elimination order to degree~$d$}
  if for every $v \in V(G)$ the set
  $S_v \coloneqq \{u \in G \mid \{u,v\} \in E(G), u \not\le v$ and
  $v \not\le u\}$ satisfies either:% \vspace{-2mm}
  \begin{enumerate}
  \itemsep=0.9pt
  \item $S_v=\emptyset$ or
  \item $v$ is $\le$-maximal, $|S_v|\le d$, and for all $u\in S_v$, we
    have $\{w \mid w < u\} = \{w \mid w < v\}$.
  \end{enumerate}
\end{definition}

A more general notion of elimination order to a class $\Cc$ was given
in the dissertation thesis of Bulian~\cite{bulian2017parameterized},
which is however not needed in this generality for our purpose.

\begin{proposition}[Proposition 4.3 in~\cite{bulian2016graph}]
  A graph $G$ satisfies $\mathrm{ed}_d(G) \le k$ if, and only if,
  there exists an elimination order to degree $d$ of depth~$k$ for~$G$.
\end{proposition}

The following lemma is easily proved by induction on $k$.

\begin{lemma}\label[lemma]{lem:cannonical-order}
  For every graph $G$ and elimination order $\leq$ to degree $d$ for
  $G$, we can compute in quadratic time an elimination order $\preceq$
  to degree $d$ for $G$, with depth not larger than the depth of
  $\leq$, and with the additional property that for every $v\in V(G)$,
  if $C,C'$ are distinct connected components of $G-V_{\preceq v}$
  (or of $G$),
  then the vertices of $C$ and $C'$ are incomparable with respect to
  $\preceq$.
\end{lemma}

Let $G$ be a graph. A graph $H$ with vertex set $\{v_1,\ldots, v_n\}$
 is a \emph{minor} of~$G$, written $H\preceq G$, if there
 are connected and pairwise vertex disjoint subgraphs
 $H_1,\ldots, H_n\subseteq G$ such that
 if $\{v_i,v_j\}\in E(H)$, then there are $w_i\in V(H_i)$ and
 $w_j\in V(H_j)$ such that $\{w_i,w_j\}\in E(G)$.
 We call the subgraph $H_i$ the \emph{branch set} of the vertex
 $v_i$ in $G$. If all $H_i$ have radius at most~$r$, then we say
 that $H$ is a \emph{depth-$r$ minor} of $G$, in symbols $H\preceq_r G$.
 If $G$ is a graph and $\mathcal{H}=\{H_1,\ldots, H_n\}$
 is a set of pairwise vertex disjoint subgraphs of $G$, then the graph
 $H$ with vertex set $\{v_1,\ldots, v_n\}$ and edges $\{v_i,v_j\}\in
 E(H)$ if and only if there is an edge between a vertex of $H_i$ and
 a vertex of $H_j$ in~$G$, the \emph{minor induced by $\mathcal{H}$}.
 If $\bigcup_{1\leq i\leq n} V(H_i)=V(G)$, we call $\mathcal{H}$ a \emph{minor model of $H$} that \emph{subsumes all
 vertices} of~$G$.

\smallskip
 We denote by $K_t$ the complete graph on $t$ vertices.
 We denote by $G_{m\times n}$ the grid with~$m$ rows and~$n$
 columns, that is, the graph with vertex set
 $\{v_{i,j}\mid 1\leq i\leq m, 1\leq j\leq n\}$ and
 edges $\{v_{i,j},v_{i',j'}\}$ for $|i-i'|+|j-j'|=1$.

 For our purpose we do not have to define the notion of treewidth
 formally. It is sufficient to note that if a graph $G$ contains an
 $n\times n$ grid as a minor, then $G$ has treewidth at least
 $n+1$ and vice versa, that large treewidth forces a large
 grid minor, as stated in the next theorem.

 \begin{theorem}[Excluded Grid Theorem]\label{thm:grid-all}
There exists a function $g$ such that for every integer
$t\geq 1$, every graph of treewidth at least $g(t)$ contains
the $t\times t$ grid as a minor. Furthermore, such a grid
minor can be computed in polynomial time.
 \end{theorem}

The theorem was first proved by Robertson and
Seymour in~\cite{robertson1986graph}. Improved bounds
and corresponding efficient
algorithms were subsequently obtained. We refer to
the work of Chuzhoy and Tan~\cite{DBLP:journals/jctb/ChuzhoyT21}
for the currently best known bounds on the function
$g$ and further pointers to the literature concerning efficient
algorithms.
For planar graphs we employ the following much better bounds.

\begin{theorem}[Planar Excluded Grid Theorem,
  see~Theorem 7.23 of~\cite{book15parameterized-algorithms}]\label{thm:grid-theorem}
  There exists an~$O(n^2)$ algorithm that, for a given $n$-vertex planar graph $G$ and integer $t$ either outputs a tree decomposition of $G$ of width at most $5t$, or construct a minor model of the $t\times t$-grid in~$G$.
\end{theorem}

The second black-box we use is Courcelle's Theorem,
stating that we can test MSO properties efficiently on graphs
of bounded treewidth. We use standard notation from logic and refer to the literature
for all undefined notation, see e.g.~\cite{libkin2013elements}.

\begin{theorem}[Courcelle's Theorem~\cite{courcelle1990monadic}]
\label{thm:courcelle}
There exists a function $f$ such that for every MSO-sentence
$\varphi$ and every $n$-vertex graph $G$ of treewidth~$t$
we can test whether
$G\models\varphi$ in time $f(|\varphi|,t)\cdot n$.
\end{theorem}

\section{The proof}

We first show that in a grid not many vertices can be affected by
$k$ recursive deletions. This is due to the following unbreakability
property of grids.

\medskip
A {\em{separation}} in a graph $G$ is a
pair of vertex subsets $A,B\subseteq V(G)$ such that \mbox{$A\cup B=V(G)$}
and there is no edge with one endpoint in $A\setminus B$ and the other
in $B\setminus A$. The {\em{order}} of the separation $(A,B)$ is the
size of its {\em{separator}} $A\cap B$.

For $q,k\in \N$, a graph $G$ is {\em{$(q,k)$-unbreakable}} if for
every separation $(A,B)$ of~$G$ of order at most~$k$, we have
$|A|\leq q$ or $|B|\leq q$.

\begin{lemma}\label[lemma]{lem-unbreak-weak}
  For all integers $k,m$, the $m\times m$-grid $G_{m\times m}$ is
  $(k^2,k)$-unbreakable.
\end{lemma}

\begin{proof}
  Let $G=G_{m\times m}$ and let $A,B$ be a separation of order at most
  $k$ in $G$.
  We show that $|A|\le k^2$ or~$|B|\le k^2$.
  The statement is trivially true if $m\le k$ as then $|V(G)|\le k^2$.
  Hence, assume that~$m>k$, and let $\z \coloneqq z_1,\ldots,z_k = A\cap B$.

  We call a row or a column of $G$ {\em untouched} if it does not intersect $\z$ and {\em touched} if it does intersect~$\z$.
  As there are at least $k+1$ rows (resp.\ columns) there is at least one
  untouched row (resp.\ column). An untouched row (resp.\ column) must be
  contained in one part of the separation.
  Assume w.l.o.g.\ that $A$ contains an untouched row. Then $A$ contains all
  untouched columns and, by symmetry, all untouched rows.
  Therefore
  every element in $B$ must be part of a touched row and a touched column, and, as there are at
  most $k$ touched rows resp.\ columns, there are at most $k^2$ such elements.
\end{proof}

Observe that when a graph with at least $3q$ vertices is $(q,k)$-unbreakable,
then removing any~$k$ vertices from $G$ leaves all but at most $q$
vertices in the same connected component. Hence, this component does
not break in the recursive process, and in particular, we delete at most
$k$ vertices from that component. On the other hand, the small part
of size at most~$q$ can possibly be completely eliminated in the recursive
process.

\begin{corollary}[of~\cref{lem-unbreak-weak}]\label[corollary]{cor-unbreak-grid-minor}
  For all integers $k,m$ with $k < m$, for all connected graphs~$G$ for which
  there exists a minor model $\mathcal{G}_{m,m} = \{H_{i,j} ~|~ 1 \le i,j\le m\}$
  that subsumes all vertices of $G$, $k$ rounds of recursive elimination on~$G$
  can delete vertices in at most $k^2$ many different branch sets.
\end{corollary}

\begin{definition}\label{def-safe-planar}
Let $G$ be a $K_5$-minor-free graph and let $k,d\in \N$.
Assume there exists a minor model $\mathcal{G}_{m,m}=\{H_{i,j}\mid
1\leq i,j\leq m\}$ (for $m\geq 4k+7$) that subsumes all
vertices of $G$ and that induces a supergraph of
the grid $G_{m,m}$. We call the branch set $H_{i,j}$ \emph{$(k,d)$-safe}
if $2k+3\leq i,j\leq m-2k-3$ and if $H_{i',j'}$ contains no vertex of degree
at least $d+1$ (in $G$) for $|i-i'|,|j-j'| \le 2k+2$.
\end{definition}

\begin{lemma}\label[lemma]{lem-irrelevant}
Let $G$ be a $K_5$-minor-free graph and let $k,d\in \N$.
Assume there exists a minor model $\mathcal{G}_{m,m}=\{H_{i,j}\mid
1\leq i,j\leq m\}$ (for $m\geq 4k+7$) that subsumes all
vertices of $G$ and that induces a supergraph of
the grid $G_{m,m}$. Assume $H_{i,j}$ is $(k,d)$-safe.
Let $a\in V(H_{i,j})$, let $B\subseteq V(G)\setminus \{a\}$
with~$|B|\leq k$ and let~$x,y\in V(G)\setminus (B\cup \{a\})$
be of degree at least $d+1$.
Then $x$ and $y$ are connected in $G-B$ if and only if
$x$ and $y$ are connected in $G-B-a$.
\end{lemma}

\begin{figure}[!b]
    \center
    \includegraphics{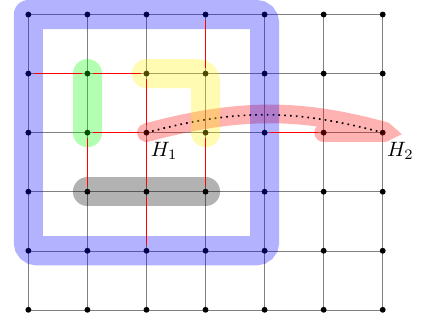}
        \caption{Construction of a $K_5$ minor as soon as a branch set (here $H_1$) connected to a branch set (here $H_2$) that is at distance more than 2 in the grid~\cite[Figure 7.10]{book15parameterized-algorithms}. The blue part marks the ``outer part'' of the border of thickness $2$, the ``inner part'' decomposes into a green, yellow and gray part.}  \label{k5minor}
  \end{figure}

\begin{proof}
Assume that $x$ and $y$ are connected in $G-B$ and let $P$ be a path
witnessing this. Assume that this path contains the vertex $a$.

\medskip
We define sets $X_\ell$ for $0\leq \ell\leq k+1$ of branch sets as follows.
Let $X_0$ be the set consisting only of $H_{i,j}$ and for $\ell\geq 1$
let $X_\ell$ be the set of all $H_{i',j'}$ with $2\ell-1 \leq |i-i'|, |j-j'|\leq 2\ell$
 that do not already belong to $X_{\ell-1}$.
The sets $X_\ell$ are the borders (of thickness $2$)
of the $(4\ell+1)\times (4\ell+1)$-subgrid around $H_{i,j}$. Observe that
the vertices $x$ and $y$ do not belong to any of the $X_\ell$, as
$H_{i,j}$ is $(k,d)$-safe by assumption. For $0\leq \ell\leq k+1$ let
$Y_\ell$ be the subgraph of $G$ induced by the vertices of $X_\ell$.

\smallskip
We claim that for $1\leq \ell\leq k+1$, the sets $Y_\ell$ are
connected sets that separate $a$ from~$x$ and analogously
$a$ from $y$. Clearly, the $Y_\ell$ are connected. Now observe that there is
no edge between a vertex of $\bigcup_{0\leq i\leq \ell-1} Y_i$
and a vertex of $G-\bigcup_{0\leq i\leq \ell} Y_i$.
The existence of such a connection would create a $K_5$ minor,
see~\cite[Figure 7.10]{book15parameterized-algorithms} or \cref{k5minor}.
Hence, any path between $a$ and $x$ (or $y$) must pass through
$Y_\ell$.

\medskip
As $|B|\leq k$, there is one $Y_\ell$ with $1\le \ell \le k+1$ that does not intersect $B$.
Let $u$ be the first time that $P$ visits $Y_\ell$ on its way from
$x$ to $a$ and let $v$ be the last time that $P$ visits~$Y_\ell$
on its way from $a$ to~$y$. As $Y_\ell$ is connected, we can
reroute the subpath between~$u$ and~$v$ through $Y_\ell$ and
thereby construct a path between $x$ and $y$ in $G-B-a$.
\end{proof}

  Note that in the proof it is important that all vertices belong to some branch set. Otherwise, we could
  have a vertex $x$ in $G$ that does not belong to any branch set while being adjacent to $H_{i,j}$ and
  the whole argumentation would fail.

\begin{corollary}\label[corollary]{cor-irrelevant}
Let $G$ be a $K_5$-minor-free graph and let $k,d\in \N$.
Assume there exists a minor model $\mathcal{G}_{m,m}=\{H_{i,j}\mid
1\leq i,j\leq m\}$ (for $m\geq 4k+7$) that subsumes all
vertices of~$G$ and that induces a supergraph of
the grid $G_{m,m}$. Assume $H_{i,j}$ is $(k,d)$-safe. Then every
vertex $a\in H_{i,j}$ is {\em irrelevant}, i.e., $G - a\in \Cc_{k,d}$
if and only if $G\in \Cc_{k,d}$.
\end{corollary}

\begin{proof}
Let $H\coloneqq G-a$. We have to prove that $H\in \Cc_{k,d}$
implies $G\in \Cc_{k,d}$. Hence, assume $H\in \Cc_{k,d}$.
Let $\le_H$ be an elimination order to degree $d$ of depth $k$ for $H$.
We also assume that~$\le_H$ satisfies the property
of~\cref{lem:cannonical-order}, that is, for every $v\in V(G)$,
if $C,C'$ are distinct connected components of $G-V_{\leq_H v}$,
then the vertices of $C$ and $C'$ are incomparable with respect to
$\leq_H$. Let $A_0$ be the connected component of $G$ containing
$a$. Note that $A_0$ may break into multiple connected components in
$H=G-a$. We show below that at most one of these connected components contains a vertex of degree excedeeing~$d$.

\medskip
For $1\le i\le k$, we define inductively:
  \begin{itemize}
    \item $m_i$ as the unique $\le_H$-minimal element of $A_{i-1}\setminus\{a\}$ (if it exists) such that there exists a vertex~$v$ with $m_i\le_H v$ and of degree at least $d+1$ in $H[A_{i-1}\setminus\{a\}]$.
    If there is no such element $m_i$, the process stops.

\smallskip
    Let us prove that there is at most one candidate for $m_i$. Assume that there are incomparable $m$ and $m'$ satisfying these conditions. This means that there are vertices $v$ and $v'$ of degree at least $d+1$ in $A_{i-1}\setminus\{a\}$ (hence of degree at least $d+1$ in $G$) with $m \le_H v$ and $m'\le_H v'$. Note that we have $m\not\le_H v'$ because $\le_H$ is a tree order.
    By \cref{lem:cannonical-order}, we have that $v,v'$ are both in $A_{i-1}\setminus\{a\}$, i.e.~in the connected component of~$a$ in $G - \{m_1,\ldots,m_{i-1}\}$ (it follows also by induction that
    $\{m_1,\ldots, m_{i-1}\}=V_{\leq_H m_{i-1}}$, hence we may
    apply the lemma). Hence, $v$ and $v'$ are connected in $G - \{m_1,\ldots,m_{i-1}\}$.
    With \cref{lem-irrelevant}, we also have that $v$ and $v'$ are connected in $G - \{a, m_1,\ldots,m_{i-1}\}$.

    We take a witness path from $v$ to $v'$. Since $m\le_H v$ and $m\not\le_H v'$, this path must contain two adjacent vertices $w$, $w'$ with $m\le_H w$ and $m\not\le_H w'$. This contradicts the fact that~$\le_H$ is an elimination order satisfying the property of \cref{lem:cannonical-order}.
    Therefore, there is at most one possible such $m_i$.

    \item $T_i := \{v\in A_{i-1} ~:~ m_i \not\le_H v\}$. These are the vertices that are disconnected from $A_i$ in $H$, but not in $G$. As explained above, they all have degree at most $d$.
    \item $A_i$ as the connected component of $a$ in $G - \{m_1,\ldots,m_i\}$. Note that again, $A_i-\{a\}$ may be a union of connected components in $H-\{m_1,\ldots,m_i\}$.
  \end{itemize}
  The processes stops after at most $k$ rounds.
  When the process stops, we have then defined $m_i$, $T_i$, and~$A_i$ up to $i=\omega$, with $\omega \le k$ and every element in $A_\omega$ has degree at most $d$ in~$H[A_\omega]$.

\medskip
  We then define the new order $\le_G$ as follows:
  \begin{itemize}
    \item for all $x,y$ other than $a$ and that are not in any of the $T_i$ nor in $A_\omega$, we have $x\le_G y$ if and only if $x\le_H y$,
    \item for all $x$ in $A_\omega \cup \bigcup\limits_{i\le\omega}T_i$, we have $m_i \le_G x$ for all $i\le \omega$. Note that $a$ is in $A_\omega$.
  \end{itemize}
  Note that all the elements in $A_\omega$, and the $T_i$'s, together with $a$ are $\le_G$-maximal.

\medskip
  We now prove that this new order is indeed an elimination order to degree $d$ of depth~$k$ for $G$.
  We have that $\le_G$ is a tree order and that it has depth at most $k$. Let us now take a vertex $b$ and study $S_b$. Recall the definition of $S_b$ from \cref{def:elim-order}. As we have two orders, we distinct $S^G_b$ from $S^H_b$.

\medskip
  First, note that for $b=m_i$, we have $S^G_{m_i} = S^H_{m_i} = \emptyset$. So we don't have to check anything.

  Then, assume that $b$ is in $A_\omega \cup \bigcup\limits_{i\le\omega}T_i$. Then $b$ is both $\le_G$ and $\le_H$-maximal. \\ Additionally,~$S_b^G=S_b^H$ if $b$ is not adjacent to $a$, and $S_b^G=S_b^H\cup \{a\}$ otherwise. Note that as $a$ is~$(k,d)$-safe, it has no neighbor of degree $d+1$, hence, if $b$ is adjacent to $a$ it has degree at most $d-1$ in $H$, and $|S_b^H|\le d-1$ so $|S_b^G|\le d$.
  For the last point of \cref{def:elim-order}, note that for any $v\in S^G_b$, $v$ also belongs to~$A_\omega \cup \bigcup\limits_{i\le\omega}T_i$, we then have $\{w~:~w <_G v\}=\{w~:~w <_G b\} = (m_i)_{i<\omega}$.

\medskip
  Finally, we look at the case where $b$ is not in $A_\omega$, is not $m_i$ nor in $T_i$ for any $i\le\omega$. In this case, we have that $b$ cannot have neighbors in $A_\omega$ nor any of the $T_i$ for $i\le\omega$. To see this, assume that there is a $v\in T_i$ neighbor to $b$. This implies that $b$ is in $A_{i-1}$, as it is connected to~$v$, the latter being in~$A_{i-1}$, which is the connected component of $a$ in $G - \{m_1,\ldots,m_{i-1}\}$. As $b\not\in T_i$, we have~$m_i \le_H b$ and~$m_i \not\le_H v$
  which contradict that $\le_H$ is an elimination order. This contradiction also holds if $b$ has a neighbor in $A_\omega$ as this would imply that $b\in A_\omega$.

\medskip
  Therefore, in this final case, $S^G_b = S^H_b$. This also holds for any neighbor of $b$. Hence for any vertex $v$ in $S^G_b$ we have that:\\
  $\{w~:~w <_G b\}=\{w~:~w <_H b\} =\{w~:~w <_H v\}=\{w~:~w <_G v\}$.

\smallskip
  To conclude, we have that $\le_G$ is indeed an elimination order to degree $d$ of depth $k$ for~$G$. This ends the proof that $a$ is irrelevant.
\end{proof}

We can now prove our main result, \cref{thm-main}.
% \alex{what follows as been slightly changed}

\begin{proof}%[Proof of \cref{thm-main}]
  Let $k,d$ be two integers and $G$ be a connected $K_5$-minor free graph.
  We set $h(k) := (4k+7)\cdot k$.
  Let~$g(\cdot)$ be the function from \cref{thm:grid-all}.

  \begin{description}
    \item First run the algorithm of \cref{thm:grid-all} with parameter $h(k)$ (or, if $G$ is planar, use the algorithm of~\cref{thm:grid-theorem}).

    \item[Case 1.] The algorithm outputs a tree-decomposition of width $g(h(k))$, or of width $5h(k)$ if $G$ is planar.
    We then use Courcelle's Theorem (\cref{thm:courcelle}) to decide
    whether $G\in \Cc_{k,d}$.

    \item[Case 2.]
    The algorithm outputs a minor model of the $h(k) \times h(k)$-grid in $G$.
    Furthermore we make sure that the set of branch sets subsumes
    all vertices. We then distinguish two cases depending on the number of
    branch sets that contain a node of degree at least $d+1$.

    \item[Case 2.1]
    If more than $k^2$ branch sets contain a node of degree at
    least $d+1$, then using \cref{cor-unbreak-grid-minor}, we conclude
    that $G\not\in \Cc_{k,d}$.

    \item[Case 2.2.] There are at most %\st{$k\cdot (2k+1)^2$}
    $k^2$ branch sets containing a vertex of degree at least $d+1$, and therefore, at most %\st{$(4k+7)^2\cdot k \cdot (2k+1)^2$}
    $k^2\cdot(4k+7)^2$ branch sets that are at distance at most $2k+3$ to a vertex of degree at least~$d+1$. As $h(k)$ is large enough, there is a $(k,d)$-safe branch set which, by \cref{cor-irrelevant}, implies the existence of an irrelevant vertex. We iteratively
     eliminate irrelevant vertices until we are in one of Case 1 or Case 2.1.
  \end{description}

  \vspace*{-9mm}
\end{proof}

We now do a quick complexity analysis of the algorithm. The algorithm of \cref{thm:grid-all} runs in time $n^{O(1)}$. Then performing either Courcelle's Theorem, finding $k^2$ branch set with vertices of degree~$d+1$, or finding an irrelevant vertex can be done in linear time, i.e.~$f(k,d)\cdot n$ for some computable function $f$. In the worst case, we might end up performing Case 2.2 up to $|G|$ many times before concluding via Case 1 or Case 2.1.
Therefore, the overall complexity is $f(k,d)\cdot n^c$ for a computable
function~$f$ and constant~$c$.

\medskip
In the slightly more restrictive case where $G$ is planar, the algorithm of \cref{thm:grid-theorem} runs in time~$O(n^2)$. This improves the complexity of the overall algorithm yielding a total running time of~$f(k,d)\cdot n^3$ for a computable
function~$f$.


\begin{thebibliography}{10}
\providecommand{\url}[1]{\texttt{#1}}
\providecommand{\urlprefix}{URL }
\expandafter\ifx\csname urlstyle\endcsname\relax
  \providecommand{\doi}[1]{doi:\discretionary{}{}{}#1}\else
  \providecommand{\doi}{doi:\discretionary{}{}{}\begingroup
  \urlstyle{rm}\Url}\fi
\providecommand{\eprint}[2][]{\url{#2}}

\bibitem{guo2004structural}
Guo J, H{\"u}ffner F, Niedermeier R.
\newblock A structural view on parameterizing problems: Distance from
  triviality.
\newblock In: Parameterized and Exact Computation, First Intl. Workshop
  (IWPEC). Springer, 2004 pp. 162--173. doi:10.1007/978-3-540-28639-4\_15.

\bibitem{gajarsky2017kernelization}
Gajarsk{\'{y}} J, Hlinen{\'{y}} P, Obdrz{\'{a}}lek J, Ordyniak S, Reidl F,
  Rossmanith P, Villaamil FS, Sikdar S.
\newblock Kernelization using structural parameters on sparse graph classes.
\newblock \emph{J. Comput. Syst. Sci.}, 2017.
\newblock \textbf{84}:219--242.  doi:10.1016/j.jcss.2016.09.002.

\bibitem{bulian2016graph}
Bulian J, Dawar A.
\newblock Graph isomorphism parameterized by elimination distance to bounded  degree.
\newblock \emph{Algorithmica}, 2016.
\newblock \textbf{75}(2):363--382.  doi:10.1007/s00453-015-0045-3.

\bibitem{DBLP:journals/siamdm/HolsKP22}
Hols EC, Kratsch S, Pieterse A.
\newblock Elimination Distances, Blocking Sets, and Kernels for Vertex Cover.
\newblock \emph{{SIAM} J. Discret. Math.}, 2022.
\newblock \textbf{36}(3):1955--1990.  doi:10.1137/20M1335285.

\bibitem{dreier2024sat}
Dreier J, Ordyniak S, Szeider S.
\newblock SAT backdoors: Depth beats size.
\newblock \emph{Journal of Computer and System Sciences}, 2024.
\newblock \textbf{142}:103520.  doi:10.1016/j.jcss.2024.103520.

\bibitem{MahlmannSV21}
M{\"{a}}hlmann N, Siebertz S, Vigny A.
\newblock Recursive Backdoors for {SAT}.
\newblock In: Intl. Symp. on Mathematical Foundations of Computer Science
  (MFCS), volume 202 of \emph{LIPIcs}. 2021 pp. 73:1--73:18.
doi:10.48550/ arXiv.2102.04707.

\bibitem{bodlaender1998rankings}
Bodlaender HL, Deogun JS, Jansen K, Kloks T, Kratsch D, M{\"u}ller H, Tuza Z.
\newblock Rankings of graphs.
\newblock \emph{{SIAM} J. Discret. Math.}, 1998.
\newblock \textbf{11}(1):168--181.

\bibitem{nevsetvril2012sparsity}
Ne{\v{s}}et{\v{r}}il J, De~Mendez PO.
\newblock Sparsity: graphs, structures, and algorithms, volume~28.
\newblock Springer Science \& Business Media, 2012.
doi:10.1007/978-3-642-27875-4.

\bibitem{reidl2014faster}
Reidl F, Rossmanith P, Villaamil FS, Sikdar S.
\newblock A Faster Parameterized Algorithm for Treedepth.
\newblock In: Intl. Coll. on Automata, Languages and Programming (ICALP),
  volume 8572 of \emph{Lecture Notes in Computer Science}. 2014 pp. 931--942.
doi:10.1007/978-3-662-43948-7\_77.

\bibitem{bulian2017fixed}
Bulian J, Dawar A.
\newblock Fixed-parameter tractable distances to sparse graph classes.
\newblock \emph{Algorithmica}, 2017.
\newblock \textbf{79}(1):139--158.  doi:10.1007/s00453-016-0235-7.

\bibitem{courcelle1990monadic}
Courcelle B.
\newblock The monadic second-order logic of graphs. {I}. Recognizable sets of finite graphs.
\newblock \emph{Information and computation}, 1990.
\newblock \textbf{85}(1):12--75.  doi:10.1016/0890-5401(90)90043-H.

\bibitem{dvovrak2013testing}
Dvo{\v{r}}{\'a}k Z, Kr{\'a}l D, Thomas R.
\newblock Testing first-order properties for subclasses of sparse graphs.
\newblock \emph{Journal of the ACM (JACM)}, 2013.
\newblock \textbf{60}(5):36.  doi:10.1145/2499483.

\bibitem{grohe2017deciding}
Grohe M, Kreutzer S, Siebertz S.
\newblock Deciding first-order properties of nowhere dense graphs.
\newblock \emph{Journal of the ACM (JACM)}, 2017.
\newblock \textbf{64}(3):17. doi:10.1145/3051095.

\bibitem{robertson1986graph}
Robertson N, Seymour PD.
\newblock Graph minors. {V}. Excluding a planar graph.
\newblock \emph{J.~Comb. Theory Ser. B}, 1986.
\newblock \textbf{41}(1):92--114.

\bibitem{robertson1995graph}
Robertson N, Seymour P.
\newblock Graph Minors. {XIII}. The Disjoint Paths Problem.
\newblock \emph{J.~Comb. Theory Ser. B}, 1995.
\newblock \textbf{63}(1):65--110.  doi:10.1006/jctb.1995.1006.

\bibitem{thilikos2012graph}
Thilikos DM.
\newblock Graph minors and parameterized algorithm design.
\newblock In: The Multivariate Algorithmic Revolution and Beyond, Springer, 2012 pp. 228--256.
    doi:10.1007/978-3-642-30891-8\_13.

\bibitem{lindermayr2020elimination}
Lindermayr A, Siebertz S, Vigny A.
\newblock Elimination Distance to Bounded Degree on Planar Graphs.
\newblock In: Intl. Symp. on Mathematical Foundations of Computer Science
  (MFCS), volume 170 of \emph{LIPIcs}. 2020 pp. 65:1--65:12.
doi:10.48550/arXiv.2007.02413.

\bibitem{DBLP:journals/siamdm/AgrawalKPRS22}
Agrawal A, Kanesh L, Panolan F, Ramanujan MS, Saurabh S.
\newblock A Fixed-Parameter Tractable Algorithm for Elimination Distance to
  Bounded Degree Graphs.
\newblock \emph{{SIAM} J. Discret. Math.}, 2022.
\newblock \textbf{36}(2):911--921. doi:10.1137/21M1396824.

\bibitem{DBLP:journals/corr/abs-2104-09950}
Agrawal A, Kanesh L, Lokshtanov D, Panolan F, Ramanujan MS, Saurabh S.
\newblock Elimination Distance to Topological-minor-free Graphs is {FPT}.
\newblock \emph{CoRR}, 2021.
\newblock \textbf{abs/2104.09950}.

\bibitem{DBLP:conf/iwpec/Agrawal020}
Agrawal A, Ramanujan MS.
\newblock On the Parameterized Complexity of Clique Elimination Distance.
\newblock In: Intl. Symp. on Parameterized and Exact Computation (IPEC), volume
  180 of \emph{LIPIcs}. 2020 pp. 1:1--1:13.

\bibitem{DBLP:journals/gc/DinerGST22}
Diner {\"{O}}Y, Giannopoulou AC, Stamoulis G, Thilikos DM.
\newblock Block Elimination Distance.
\newblock \emph{Graphs Comb.}, 2022.
\newblock \textbf{38}(5):133.  doi:10.1007/s00373-022-02513-y.

\bibitem{DBLP:journals/tocl/FominGT22}
Fomin FV, Golovach PA, Thilikos DM.
\newblock Parameterized Complexity of Elimination Distance to First-Order Logic  Properties.
\newblock \emph{{ACM} Trans. Comput. Log.}, 2022.
\newblock \textbf{23}(3):17:1--17:35.  doi:10.1145/3517129.

\bibitem{DBLP:conf/icalp/PilipczukSSTV22}
Pilipczuk M, Schirrmacher N, Siebertz S, Torunczyk S, Vigny A.
\newblock Algorithms and Data Structures for First-Order Logic with
  Connectivity Under Vertex Failures.
\newblock In: Intl. Coll. on Automata, Languages and Programming (ICALP),
  volume 229 of \emph{LIPIcs}. 2022 pp. 102:1--102:18.
  doi:10.4230/LIPIcs.ICALP.2022.102.

\bibitem{DBLP:journals/tocl/SchirrmacherSV23}
Schirrmacher N, Siebertz S, Vigny A.
\newblock First-order Logic with Connectivity Operators.
\newblock \emph{{ACM} Trans. Comput. Log.}, 2023.
\newblock \textbf{24}(30):1--23.  doi:10.1145/3595922.

\bibitem{bulian2017parameterized}
Bulian J.
\newblock Parameterized complexity of distances to sparse graph classes.
\newblock Technical report, University of Cambridge, Computer Laboratory, 2017.
doi:10.48456/tr-902.

\bibitem{DBLP:journals/jctb/ChuzhoyT21}
Chuzhoy J, Tan Z.
\newblock Towards tight(er) bounds for the Excluded Grid Theorem.
\newblock \emph{J. Comb. Theory, Ser. {B}}, 2021.
\newblock \textbf{146}:219--265.  doi:10.1016/j.jctb.2020.09.010.

\bibitem{book15parameterized-algorithms}
Cygan M, Fomin FV, Kowalik L, Lokshtanov D, Marx D, Pilipczuk M, Pilipczuk M,
  Saurabh S.
\newblock Parameterized Algorithms.
\newblock Springer, 2015.  doi:10.1007/978-3-319-21275-3.

\bibitem{libkin2013elements}
Libkin L.
\newblock Elements of finite model theory.
\newblock Springer Science \& Business Media, 2013.
ISBN:3540212027, 9783540212027.
\end{thebibliography}
\end{document}